%% file: main.tex
\documentclass[article]{eptcs}
\providecommand{\event}{ITRS 2012} 
\usepackage{breakurl}             
\usepackage{color} 
\usepackage{cmll}
\usepackage{latexsym}
\usepackage{amssymb,amsmath,amsbsy}
\usepackage{graphicx}
\usepackage{stmaryrd}
\usepackage{pstricks}
\usepackage[all]{xy}
\usepackage{enumerate}
\usepackage{hyperref}
\usepackage{mathrsfs}
\usepackage{amsthm}
\usepackage{verbatim}
\usepackage{wasysym}
\usepackage{proof}

\title{Bounding normalization time through intersection types}
\author{Erika De Benedetti\footnote{The first author is supported by the MIUR Project IPODS.} \qquad \qquad Simona Ronchi Della Rocca
\institute{Dipartimento di Informatica\\
Universit\`a di Torino, Italy
\email{debenede@di.unito.it \qquad ronchi@di.unito.it}
}
}
\def\titlerunning{Bounding normalization time through intersection types}
\def\authorrunning{E. De Benedetti \& S. Ronchi Della Rocca}
\input{macro}

\begin{document}
\maketitle

\section{Introduction}
Intersection types were originally introduced as idempotent, i.e., modulo the equivalence $\sigma \wedge \sigma =\sigma$.
In fact, they have been used essentially for semantic purposes, for building filter models for $\lambda$-calculus,
where the interpretation of types as properties of terms induces naturally the idempotence property.\\ 
Recently it has been observed that, when dropping idempotency, intersection types can be used for reasoning about the complexity of $\beta$-reduction.
Some results have been already obtained along this line. Terui \cite{TeruiInt} designed a system assigning non-idempotent intersection types to $\lambda$-calculus, which can type all and only the strongly normalizing terms, and such that the size of any derivation with subject $\M$ is bigger than the size of every term in the $\beta$-reduction sequence from $\M$ to its normal form. 
This property can be used for computing a bound of every normalizing $\beta$-reduction sequence starting from $\M$. A more precise result in this direction has been obtained by Lengrand \cite{bernadetleng11}, who gave a precise measure of the number of $\beta$-reduction steps. Namely he designed a type assignment system, where intersection is considered without idempotency, and defined the notions of measure of derivation and of principal derivation for a given term. Then he proved that the measure of a principal derivation of a type for a normalizing term $\M$ corresponds to the maximal length of a normalizing $\beta$-reduction sequence for $\M$. \\
In this line, we go one step forward, and use intersection types without neither idempotence nor associativity to express the functional dependence of the length of a normalizing $\beta$-reduction sequence from a term $\M$ on the size of $\M$ itself . 
In order to obtain such a result, we take inspiration from the system \STA\ of Gaboardi and Ronchi Della Rocca \cite{GaboardiRonchi07csl}, in its turn inspired by the Soft Linear Logic of
Lafont \cite{Lafont04}, which characterizes the polynomial time computations. The resulting system allow us to give a bound on the number of steps necessary to reduce a normalizing term $\M$ to its normal form, in the form $|\M|^{d+1}$, where $|\M|$ is the size of the term, and $d$ is a measure depending on the type derivation for it (the {\em depth}). Since for every normalizing term there is a type derivation with minimal depth,  this bound does not depend on a particular derivation.
A preliminary type assignment of this kind has been described in \cite{DeBen11}.

\medskip
Some type assignment systems without idempotency have been already studied in the literature, for various purposes.
Kfoury and Wells, in \cite{kfouryWells04} used non-idempotent intersection
in order to formalize a type inference semi-algorithm, whose complexity has been studied in 
\cite{MairsonNeergaard04}. Kfoury, in \cite{kfoury00}, connected non idempotent intersection types with linear $\beta$-reduction.
Recently non idempotent intersection types have been used by Pagani and Ronchi Della Rocca for characterizing the
solvability in the resource $\lambda$-calculus \cite{pagani10fossacs, PaganiRonchi:2010FI}.
In \cite{DiGia08} the game semantics of a typed $\lambda$-calculus has been described in logical form using an intersection type 
assignment system where the intersection is not idempotent neither commutative nor associative. 
Some complexity results have been obtained by De Carvalho in \cite{deCarvalho09CORR}, using a $\lambda$-algebra induced by non idempotent types. Recently a logical description of relational model of $\lambda$-calculus \cite{bucciarelli07csl} has been designed, through a non-idempotent type assigment system \cite{paolini12draft}.

\section{System \STI}
\input{sys}

\input{ex}
\section{Normalization bound}
\input{measure}

\bibliographystyle{eptcs}
\small
\bibliography{main}

\end{document}

%% file: macro.tex
\newtheorem{definition}{Definition}
\newtheorem{theorem}{Theorem}
\newtheorem{property}{Property}
\newtheorem{example}{Example}
\newtheorem{remark}{Remark}
\newtheorem{lemma}{Lemma}

{\bf}{\it}
\newcommand{\lin}{\rightarrow}

\newcommand{\dem}{\triangleright}
\newcommand{\der}{\vdash}

\newcommand{\SLL}{{\tt SLL}}
\newcommand{\STA}{{\tt STA}}

\newcommand{\STI}{{\tt STI}}

\newcommand{\FV}[1]{{\rm FV}(#1)}
\newcommand{\dom}[1]{\mbox{dom}(#1)}
\newcommand{\RK}[1]{{\tt rk}(#1)}
\newcommand{\D}[1]{{\tt d}(#1)}
\newcommand{\W}[2]{{\tt W}(#1, #2)}

\newcommand{\lam}{\lambda}
\newcommand{\M}{{\tt M}}
\newcommand{\N}{{\tt N}}

\newcommand{\Q}{{\tt Q}}

\newcommand{\x}{{\tt x}}
\newcommand{\y}{{\tt y}}
\newcommand{\z}{{\tt z}}

\newcommand{\redbeta}{\underset{\beta}{\longrightarrow}}
\newcommand{\redbetas}{\overset{*}{\underset{\beta}{\longrightarrow}}}

\newcommand{\sub}[3]{#1 [#2/#3]}

\newcommand{\A}{{\tt A}}
\newcommand{\B}{{\tt B}}
\newcommand{\C}{{\tt C}}
\newcommand{\tvar}{{\tt a}}

\newcommand{\stra}[2]{{#1}_{1} \wedge ... \wedge {#1}_{#2}}


%% file: sys.tex
We start by introducing \STI\ (Soft Type assignment with Intersection), a type assignment system for $\lam$-calculus assigning to $\lambda$-terms non-idempotent and not associative intersection types. The system assigns types to all and only strongly normalizing terms.

\begin{definition}
\begin{enumerate}[i)]
\item Terms of $\lam$-calculus are defined by the following grammar:
$$\M ::= \x \mid \M\M \mid \lam \x. \M$$
where $\x$ ranges over a countable set ot variables. The symbol $\equiv$ denotes the syntactical equality modulo renaming of bound variables.
\item The reduction relation $\redbeta$ is the contextual closure of the rule $(\lam \x. \M) \N \rightarrow \M[\N/\x]$,
where $\M[\N/\x]$ is the capture-free substitution of $\N$ to all the free occurrences of $\x$ in $\M$.  $\redbetas$ is the reflexive and transitive closure of $\redbeta$.

\item A term $\M$ is an instance of $\N$ if it is obtained from $\N$ by renaming a subset of its free variables with a unique fresh name.

\item The set of \STI\ types is defined as follows:
\begin{align*}
\A ::= & \tvar \mid \sigma \lin \A & \mbox{(linear types)}\\
\sigma ::= & \A \mid \underbrace{ \sigma \wedge ... \wedge \sigma}_n \quad (n > 1) & \mbox{(intersection types)}
\end{align*}
where $\tvar$ ranges over a countable set of type variables. Linear types are ranged over by $\A, \B, \C$, intersection types by 
$\sigma, \tau, \rho$.  The connective $\wedge$ is commutative, but it is not idempotent nor associative.

The number of elements of a type is defined inductively as $l(\A) = 1$, $l(\stra{\sigma}{n})= l(\sigma_{1}) + ... + l(\sigma_{n})$.

\item A context is a finite set of assumptions of the shape $\x: \sigma$, where $\x$ is a variable and $\sigma$ is a type. Variables in a context are all distinct, and
contexts are ranged over by $\Gamma, \Delta$. $dom(\Gamma)$ is the set $\{x \mid x:\sigma \in \Gamma \}$.
The intersection of contexts is given by
$$\Gamma \wedge \Delta = \{ \x : \sigma \mid \x : \sigma \in \Gamma, \x \not\in \dom{\Delta} \} \cup \{ \x : \tau \mid \x : \tau \in \Delta, \x \not\in \dom{\Gamma} \} \cup \{ \x : \sigma \wedge \tau \mid \x:\sigma \in \Gamma, \x:\tau \in \Delta \}$$
while $\Gamma, \Delta$ represents the union of sets $\Gamma$ and $\Delta$, provided that $\Gamma \# \Delta$, i.e. $\dom{\Gamma} \cap \dom{\Delta} = \emptyset$.

\item The system \STI\ proves sequents of the shape $\Gamma \der \M : \sigma$,
where $\Gamma$ is a context, $\M$ is a term of $\lambda$-calculus, and $\sigma$ is a type. The rules are given in Table $\ref{tab:rules}$.

\item Derivations are denoted by $\Pi, \Sigma$. $\Pi \dem \Gamma \der \M: \sigma$ denotes a derivation $\Pi$ with conclusion 
$\Gamma \der \M: \sigma$.

\end{enumerate}
\end{definition}

\begin{table}
\small
\begin{center}
\begin{tabular}{cc}
\hline
\\
$\infer[(Ax)]
     {\x: \A \der \x: \A}
     {}
\qquad
\infer[(w)]
     {\Gamma, \x: \A \der \M : \sigma}
     {\Gamma \der \M : \sigma & \x \notin \dom{\Gamma}}
$
\\
\\
$\infer[(\lin I)]
     {\Gamma \der \lambda \x. \M: \sigma \lin \A}
     {\Gamma, \x: \sigma  \der \M: \A}
\qquad
\infer[(\lin E)]
     {\Gamma, \Delta \der \M\N: \A}
     {\Gamma \der \M: \sigma \lin \A & \Delta \der \N: \sigma & \Gamma \# \Delta}$
     \\
     \\
     $\infer[(\wedge_n)]
     {\bigwedge_{i=1}^n \Gamma_i \der \M: \stra{\sigma}{n}}
     {\Gamma_{1} \der \M: \sigma_{1} & ... & \Gamma_{n} \der \M: \sigma_{n} & n > 1}$
     \\
     \\
     $\infer[(m)]
     {\Gamma, \x:  \stra{\sigma}{n} \der \M [\x/\x_{1}, ...\ , \x/\x_{n}]: \tau}
     {\Gamma, \x_{1} : \sigma_{1}, ..., \x_{n}: \sigma_{n} \der \M: \tau}$
\\
\\
\hline

\end{tabular}
\end{center}
\caption{The type assignment system}
\label{tab:rules}
\normalsize
\end{table} 

Some comments are in order. Since the condition on contexts in rule $(\lin E)$, terms are built in a linear form, and an explicit multiplexor rule is present (rule $(m)$). This allows to control the number of (multiple) contractions, which is responsible for the growth of the reduction time. The counterpart of the contraction on the right side of a derivation is the rule $(\wedge_n)$, which is parametric in $n$. In doing this, we were inspired by the Soft Linear Logic of Lafont.

Let us define {\em constructive} the rules, which contribute in building the subject, i.e., either $(Ax)$, or $(\lin I)$ or $(\lin E)$). 

\begin{definition}[Intersection trees]
Let $(\delta)$ be a (possibly empty) sequence of applications of rules $(w)$ and $(m)$. An {\em intersection tree} is a maximal (sub)proof of the shape defined inductively in the following way:
\begin{itemize}
\item Let the last rule of $\Sigma$ be a constructive rule . Then
\small
$$
\infer=[(\delta)]{\Gamma \der M:\sigma }{\Sigma}
$$
\normalsize
is an empty intersection tree, with conclusion $\Gamma \der M:\sigma$ and one leaf $\Sigma$.
\item 
If $\Sigma_i$ is a (possibly empty) intersection tree ($1\leq i\leq n$), then
\small
$$
\infer=[(\delta)]{\Gamma \der \M':\sigma }{\infer[(\wedge_n)] {\bigwedge_{i=1}^n \Gamma_i \der \M: \stra{\sigma}{n}}
{\Sigma_i\dem\Gamma_{i} \der \M: \sigma_{i} \quad (1 \leq i \leq n)}}
$$
\normalsize
is an intersection tree, with conclusion $\Gamma\der \M':\sigma$, where $\M'$ is an instance of $\M$, $\Gamma$ is a contraction of $\bigwedge_{i=1}^n \Gamma_i $, and its leaves are the leaves of all the $\Sigma_i$.
\end{itemize}
\end{definition}

Since the $(\wedge_n)$ rule is the only rule building an intersection type on the right of the turnstile symbol, it is possible to state the following, which is a key property for proving the normalization bound.

\begin{property}[Subject with intersection type]
\label{prop:stra}
Let $\Pi \dem \Gamma \der \M: \stra{\sigma}{m}$ with $m > 1$. Then $\Pi$ ends with a non empty intersection tree. 
\end{property}

\begin{proof}
By induction on the shape of $\Pi$.
If the last applied rule is $(\wedge_n)$, then the statement is trivially true and $\delta$ is the empty sequence.
Otherwise, the derivation needs to contain at least one application of rule $(\wedge_n)$, with subject $\M'$, such that
$\M$ is an instance of $\M'$. Then this application can be followed only by $\delta$ of rules, which can contain only applications of rule $(w)$ or rule $(m)$.
\end{proof}

The substitution property holds for terms having disjoint free variables sets.

\begin{lemma}[Substitution]
\label{lem:subs}
Let $\Pi \dem \Gamma, \x : \sigma \der \M : \tau$, $\Sigma \dem \Delta \der \N : \sigma$, $\Gamma \# \Delta$
and $\x \not \in dom(\Delta)$.

Then there exists $S(\Sigma, \Pi)$ such that $S(\Sigma, \Pi) \dem \Gamma, \Delta \der \sub{\M}{\N}{\x} : \tau$.
\end{lemma}

\begin{proof} By induction on the shape of $\Pi$. The proof is trivial except for the cases of $(w)$, $(\wedge_n)$ or $(m)$.

If $(w)$ is the last applied rule introducing a variable $\y \not= \x$, then the proof follows by induction. Otherwise, let $\Pi$ be the proof
\small
$$\infer[(w)]{\Pi \dem \Gamma, \x: \A \der \M: \sigma}{\Pi' \dem \Gamma \der \M: \sigma & \x \notin \dom{\Gamma}}$$
\normalsize
and let $\Sigma \dem \Delta \der \N : \A$. 
If $\Delta$ contains only bindings of variables to linear types, then $S(\Sigma, \Pi)$ is the proof
\small
$$\infer=[(w)]{S(\Sigma, \Pi) \dem \Gamma, \Delta \der \M: \sigma}{\Pi' \dem \Gamma \der \M: \sigma}$$
\normalsize
Otherwise, let us assume, without loss of generality, $\Delta = \Delta', \y : \tau$ such that $\A_{1}, ...\ , \A_{n}$ are the elements of $\tau$, and let $\Delta'$ contain only bindings of variables to linear types. Then $S(\Sigma, \Pi)$ is the proof
\small
$$\infer=[(w)]{\infer=[(m)]{S(\Sigma, \Pi) \dem \Gamma, \Delta', \y: \tau \der \M : \sigma}{\Gamma, \Delta', \y_1: \A_1, ...\ , \y_n: \A_n \der \M: \sigma}}{\Pi' \dem \Gamma \der \M: \sigma}$$
\normalsize
where the sequence of applications of rule $(m)$ is constructiong $\tau$.

If the last applied rule is $(\wedge_n)$, with $n > 1$, then $\Pi$ is of the shape 
\small
$$\infer[(\wedge_{n})]{\Gamma, \x : \stra{\sigma}{n} \der \M: \stra{\tau}{n}}{\Pi_{1} \dem \Gamma_{1}, \x: \sigma_{1} \der \M: \tau_{1} & ... & \Pi_{n} \dem \Gamma_{n}, \x: \sigma_{n} \der \M: \tau_{n}}$$ 
\normalsize
By Property \ref{prop:stra}, $\Sigma$ is of the shape
\small
$$\infer[(\wedge_{n})]{\infer=[(\delta)]{\Sigma \dem \Delta \der \N: \stra{\sigma}{n}}{\Delta' \der \N': \stra{\sigma}{n}}}{\Sigma_{1} \dem \Delta_{1} \der \N': \sigma_{1} & ... & \Sigma_n \dem \Delta_n \der \N': \sigma_{n}}$$
\normalsize
where $\delta$ is a sequence of applications of $(w)$ and $(m)$ rules, and $\N$ is an instance of $\N'$.

By inductive hypothesis $S(\Sigma_i, \Pi_i) \dem \Gamma_i, \Delta_i \der \sub{\M}{\N'}{\x} : \tau_i$, since $\Gamma \# \Delta$ implies $\Gamma_i \# \Delta_i$ for all $i$ ,
so $S(\Sigma, \Pi)$ is given by
\small
$$\infer=[(\delta)]{\Gamma, \Delta \der \sub{\M}{\N}{\x}: \stra{\tau}{n}}{\infer[(\wedge_n)]{\Gamma, \Delta' \der \sub{\M}{\N'}{\x}: \stra{\tau}{n}}{S(\Sigma_{1}, \Pi_{1}) \dem \Gamma_{1}, \Delta_{1} \der \sub{\M}{\N'}{\x}: \tau_{1} & ... & S(\Sigma_{n}, \Pi_{n}) \dem \Gamma_{n}, \Delta_{n} \der \sub{\M}{\N'}{\x}: \tau_{n}}}$$
\normalsize
If the last applied rule is $(m)$, then $\Pi$ is of the shape 
\small
$$\infer[(m)]
     {\Gamma, \x: \stra{\sigma}{n} \der \M[\x/\x_{1}, ...\ , \x/\x_{n}]: \tau}
     {\Pi' \dem \Gamma, \x_1: \sigma_{1}, \dots, \x_n: \sigma_{n} \der \M: \tau}$$
\normalsize
Exactly as in the previous case, we can apply Property \ref{prop:stra} to $\Sigma$, thus obtaining $\Sigma_{i} \dem \Delta_{i} \der \N' : \sigma_{i}$, for $1 \leq i \leq n$.
Also, we must rename the variables in $\Delta_{i}$, so that 
we actually get proofs $\Sigma'_{i} \dem \Delta'_{i} \der \N'_{i}: \sigma_{i}$ where $\N'_{i}$ is an instance of $\N'$ and all $\dom{\Delta'_{i}}$ are disjoint from each other; this is not a trouble as we will be able to recover $\Delta'$ and $\N'$ easily by a suitable sequence $\rho$ of applications of $(m)$ rules.

By induction we can now build

$$S(\Sigma'_1, \Pi') \dem \Gamma, \Delta'_1, \x_2: \sigma_2, ... , \x_n: \sigma_n \der \sub{\M}{\N'_{1}}{\x_{1}}: \tau$$
$$S(\Sigma'_2, S(\Sigma'_1, \Pi')) \dem \Gamma, \Delta'_1, \Delta'_2, \x_3: \sigma_3, ... , \x_n: \sigma_n \der \M[\N'_{1}/\x_{1}, ...\ , \N'_{2}/\x_{2}]: \tau$$
$$\vdots$$
$$S(\Sigma'_n, S(\Sigma'_{n-1}, ... S(\Sigma'_1, \Pi') ... )) \dem \Gamma, \Delta'_{1}, ... , \Delta'_{n} \der \M[\N'_{1}/\x_{1}, ...\ , \N'_{n}/\x_{n}]: \tau$$

and by applying sequences $\rho$ and $\delta$ of rule $(m)$, we get the desired proof
\small
$$\infer=[(\delta)]{S(\Sigma, \Pi) \dem \Gamma, \Delta \der \M[\N/\x_{1}, ...\ , \N/\x_{n}]: \tau}{
\infer=[(\rho)]{\Gamma, \Delta' \der \M[\N'/\x_{1}, ...\ , \N'/\x_{n}]: \tau}
{\Gamma, \Delta'_{1}, ... , \Delta'_{n} \der \M[\N'_{1}/\x_{1}, ...\ , \N'_{n}/\x_{n}]: \tau}}$$
\normalsize
\end{proof}

The substitution property is sufficient for proving the subject reduction property, but we need to take into account that one step of $\beta$-reduction on the subject can be matched by a \textsl{set} of $n \geq 1$ parallel simplification steps in the underlying derivation, corresponding to reducing virtual copies of the same redex having different types.

\begin{property}[Subject reduction]
\label{prop:subjred}
$\Pi \dem \Gamma \der \M :\sigma$ and $\M \redbeta \M'$ implies $\Pi' \dem \Gamma \der \M': \sigma$.
\end{property}

\begin{proof}
$\M \redbeta \M'$ means $\M=C[(\lambda \x.\Q) \N]$ and $\M'=C[\Q[\N/\x]]$, for some context $C[.]$. The proof is by induction on $C[.]$. 
Let us consider just the base case in which $C[.]=[.]$, i.e., $\M=(\lambda \x.\Q)\N$. Then the most difficult case is when $\Pi$ ends by a non empty intersection tree. Note that the shape of $\M$ implies each leaf $\Pi_i$ of the intersection tree be of the shape:  
\small
$$\infer[(\lin E)]{\Gamma_i, \Delta_i \der (\lambda \x. \Q_i) \N_i: \A_i}{\infer=[(\delta_i)]{\Gamma_i \der \lambda \x. \Q_i : \sigma_i \lin \A}{\infer[(\lin I)]{\Gamma'_i \der \lambda \x. \Q'_i : \sigma_i \lin \A_i}{\Sigma'_i\dem \Gamma'_i, \x: \sigma_i \der \Q'_i : \A_i}} & \Sigma''_i \dem \Delta_i \der \N_i : \sigma_i & \Gamma_i \# \Delta_i}$$
\normalsize
where $1\leq i \leq n$, for some $n >1$, $(\lambda \x.\Q)\N$ is an instance of $(\lambda \x. \Q_i) \N_i$, and $\delta_i$ is a (possibly empty) sequence of applications of $(w)$ and $(m)$ rules.  Since all $(m)$ rules in $\delta_i$ deal with variables in $\dom{\Gamma}$, sequence $\delta_i$ can be delayed to obtain the proof
\small
$$\infer=[(\delta_i)]{\Gamma_i, \Delta_i \der  (\lambda \x. \Q_i) \N_i: \A_i}{\infer[(\lin E)]{\Gamma'_i, \Delta_i \der (\lambda \x. \Q'_i) \N_i: \A_i}{\infer[(\lin I)]{\Gamma'_i \der \lambda \x. \Q'_i : \sigma_i \lin \A_i}{\Sigma'_i \dem \Gamma'_i, \x: \sigma_i \der \Q'_i : \A} & \Sigma_i'' \dem \Delta_i \der \N_i : \sigma_i}}$$
\normalsize
By Lemma \ref{lem:subs}, there are proofs $S(\Sigma''_i,\Sigma'_i) \dem \Gamma_i, \Delta_i \der
\Q_i[\N_i/x]$, and then the result is obtained by replacing the leafs $\Pi_i$ of the intersection tree by
$S(\Sigma''_i,\Sigma'_i)$ ($1\leq i \leq n$). 


\end{proof}

Moreover the system is strongly normalizing. Formally:

\begin{property}[Strong normalization]
\label{prop:sn}
$\Pi \dem \Gamma \der \M :\sigma$ if and only if $\M$ is strongly normalizing.
\end{property}

For the right implication, the proof is obtained in the next section by observing that the measure of $\Pi$ decreases with each reduction step, and this does not depend on any particular strategy. As for the left implication, the proof can be obtained by adapting Neergaard's proof \cite{Neergaard05} to system \STI. In fact, Neergaard proved the strong normalization property for a system with rigid intersection types, i.e. intersection without commutativity, associativity nor idempotency.

%% file: ex.tex
\begin{example}
\small Here we will show an example of a derivation in \STI, aiming to clarify the behaviour on the subject reduction in the case of a non-empy intersection tree.
Let

$$
\infer[(\lin E)]{\Sigma_{1} \dem \z: \A \der (\lam \y. \y) \z: \A}
{\infer[(\lin I)]{\der \lam \y. \y: \A \lin \A}
{\infer[(Ax)]{\y: \A \der \y: \A}{}}
&&
\infer[(Ax)]{\z: \A \der \z: \A}{}}
\hspace{10px}
\mbox{ and }
\hspace{10px}
\infer[(\lin E)]{\Sigma_{2} \dem \z: \tvar \der (\lam \y. \y) \z: \tvar}
{\infer[(\lin I)]{\der \lam \y. \y: \tvar \lin \tvar}
{\infer[(Ax)]{\y: \tvar \der \y: \tvar}{}}
&&
\infer[(Ax)]{\z: \tvar \der \z: \tvar}{}}
$$

where $\A = \tvar \lin \tvar$.

We want to reduce the term $(\lam \x. \x \x) ((\lam \y. \y) \z)$ to normal form; the derivation is the following:

$$\infer[(\lin E)]{\z: \A \wedge \tvar \der (\lam \x. \x \x)((\lam \y. \y) \z): \tvar}
{
\infer[(\lin I)]{\der \lam \x. \x \x: (\A \wedge \tvar) \lin \tvar}
{\infer[(m)]{\x: \A \wedge \tvar \der \x \x: \tvar}
{\infer[(\lin E)]{\x_{1}: \A, \x_{2}: \tvar \der \x_{1} \x_{2}: \tvar}
{\infer[(Ax)]{\x_{1}: \A \der \x_{1}: \A}{} & \infer[(Ax)]{\x_{2}: \tvar \der \x_{2}: \tvar}{}}}}
&
\infer[(\wedge_{2})]{\Sigma \dem \z: \A \wedge \tvar \der (\lam \y. \y) \z: \A \wedge \tvar}
{
\Sigma_{1} \dem \z: \A \der \mbox{\fbox{$(\lam \y. \y) \z$}}: \A
&
\Sigma_{2} \dem \z: \tvar \der \mbox{\fbox{$(\lam \y. \y) \z$}}: \tvar
}}$$

Notice that, since $\Sigma$ ends by a non empty intersection tree, there are two ``virtual`` copies of the same redex; therefore, if we reduce the redex $(\lam \y. \y) \z$, we get the following derivation:

$$\infer[(\lin E)]{\Pi \dem \z: \A \wedge \tvar \der (\lam \x. \x \x) \z: \tvar}
{
\infer[(\lin I)]{\der \lam \x. \x \x: (\A \wedge \tvar) \lin \tvar}
{\infer[(m)]{\x: \A \wedge \tvar \der \x \x: \tvar}
{\infer[(\lin E)]{\x_{1}: \A, \x_{2}: \tvar \der \x_{1} \x_{2}: \tvar}
{\infer[(Ax)]{\x_{1}: \A \der \x_{1}: \A}{} & \infer[(Ax)]{\x_{2}: \tvar \der \x_{2}: \tvar}{}}}}
&
\infer[(\wedge_{2})]{\z: \A \wedge \tvar \der \z: \A \wedge \tvar}
{
\infer[(Ax)]{\z: \A \der \z: \A}{}
&
\infer[(Ax)]{\z: \tvar \der \z: \tvar}{}
}}$$

where both the redexes of $\Sigma_{1}$ and $\Sigma_{2}$ have been reduced.

Finally, we reduce $(\lam \x. \x \x) \z$ (easy, as $\Pi$ ends with an empty intersection tree), obtaining the proof

$$\infer[(m)]{\z: \A \wedge \tvar \der \z \z: \tvar}
{\infer[(\lin E)]{\z_{1}: \A, \z_{2}: \tvar \der \z_{1} \z_{2}: \tvar}
{\infer[(Ax)]{\z_{1}: \A \der \z_{1}: \A}{} & \infer[(Ax)]{\z_{2}: \tvar \der \z_{2}: \tvar}{}}}$$

Notice that, as explained in the proof for Lemma \ref{lem:subs}, the premises of rule $(\wedge_{2})$ need to be rewritten in the substitution so that their contexts are disjoint; the original context is then recovered by a suitable sequence of $(m)$ rules.
\end{example}

%% file: measure.tex
In computing the normalization bound, we take inspiration from \SLL\  \cite{Lafont04} and \cite{GaboardiRonchi07csl}, but taking into account the mismatch between proof simplification and $\beta$-reduction.
So here we do not use the derivation as reduction machine, but rather as a tool for computing the number of reduction steps.

To do so, we first introduce a few necessary definitions of measures.

\begin{definition}[Measures]
\label{def:measures}
\end{definition}

\begin{enumerate}[i)]

\item The \textbf{size} $|\Pi|$ of a proof $\Pi$ is defined inductively as follows:
\begin{itemize}
\item if the last rule of $\Pi$ is the axiom rule, then $|\Pi| = 1$;
\item if the last rule of $\Pi$ is a rule with $n$ premises $\Pi_i$, then $|\Pi| = \left( \sum_{i=1}^n |\Pi_i| \right) + 1$.
\end{itemize}

\item The \textbf{size} $|\M|$ of a term $\M$ is defined inductively as follows:
$$|\x| = 1; \qquad |\lambda \x. \M| = |\M| + 1; \qquad |\M\N| = |\M| + |\N| + 1.$$

\item The \textbf{rank} of a multiplexor
\small
$$\infer[(m)]
     {\Gamma, \x: \stra{\sigma}{n} \der \M[\x_i \mapsto \x]_{i=1}^n: \tau}
     {\Gamma, \x_1: \sigma_1, ..., \x_n: \sigma_n \der \M: \tau}$$
\normalsize
is the number $k \leq n$ of variables $\x_i$ such that $\x_i \in \FV{\M}$. Let $r$ be the maximum rank of a rule $(m)$ in $\Pi$. The rank $\RK{\Pi}$ of $\Pi$ is the maximum between $1$ and $r$.

\item The \textbf{degree} of a proof $\Pi$, denoted by $\D{\Pi}$, is the maximal nesting of applications of the $(\wedge_n)$ rule in $\Pi$, i.e. the maximal number of applications of the $(\wedge_n)$ rule in a path connecting the conclusion and one axiom of $\Pi$.

\item The \textbf{weight} $\W{\Pi}{r}$ of $\Pi$ with respect to $r$ is defined inductively as follows:
\begin{itemize}
\item if $(Ax)$ is the last applied rule, then $\W{\Pi}{r} = 1$;
\item if $(\lin I)$ is the last applied rule and $\Sigma$ is the premise of the rule, then $\W{\Pi}{r} = \W{\Sigma}{r} + 1$;
\item if $(\lin E)$ is the last applied rule and $\Sigma_1, \Sigma_2$ are the premises of the rule, then $\W{\Pi}{r} = \W{\Sigma_1}{r} + \W{\Sigma_2}{r} + 1$;
\item if $(\wedge_n)$ is the last applied rule and $\Sigma_1, ..., \Sigma_n$ are the premises of the rule, then $\W{\Pi}{r} = r \cdot \max_{i=1}^n \W{\Sigma_i}{r}$;
\item if either $(w)$ or $(m)$ is the last applied rule and $\Sigma$ is the unique premise derivation, then $\W{\Pi}{r} = \W{\Sigma}{r}$.
\end{itemize}

\end{enumerate}

The previously introduced measures are related to each other as shown explicitly by the following lemma:

\begin{lemma}
Let $\Pi \dem \Gamma \der \M: \sigma$. Then:
\begin{enumerate}[i)]
\item $\RK{\Pi} \leq |\M| \leq |\Pi|$.
\item $\W{\Pi}{r} \leq r^{\D{\Pi}} \cdot \W{\Pi}{1}$.
\item $\W{\Pi}{1} = |\M|$.
\end{enumerate}
\label{lem:3m}
\end{lemma}

\begin{proof}
The proofs are given by induction on the shape of $\Pi$.
\begin{enumerate}[i)]
\item The most interesting case is for $\Pi$ of the shape
\small
$$\infer[(m)]{\Pi \dem \Gamma, \x: \stra{\tau}{n} \der \M[\x/\x_{1}, ...\ , \x/\x_{n}]: \sigma}{\Sigma \dem \Gamma, \x_1: \tau_{1}, ... , \x_n: \tau_{n} \der \M: \sigma}$$
\normalsize
By inductive hypothesis, $\RK{\Sigma} \leq |\M| \leq |\Sigma|$.

Let $k \leq n$ be the number of variables in $\{ \x_1, ... , \x_n\} \cap \FV{\M}$. By Definition \ref{def:measures}, $\RK{\Pi} = \max\{\RK{\Sigma}, k\}$, $k \leq |\M[\x/\x_{1}, ...\ , \x/\x_{n}]| = |\M|$ and $|\Pi| = |\Sigma|  + 1$, therefore
\begin{itemize}
\item if $\max\{\RK{\Sigma}, k)\} = \RK{\Sigma}$, then $\RK{\Pi} = \RK{\Sigma} \leq |\M[\x/\x_{1}, ...\ , \x/\x_{n}]| \leq |\Sigma| + 1$
\item if $\max\{\RK{\Sigma}, k)\} =k$, then $\RK{\Pi} = k \leq |\M[\x/\x_{1}, ...\ , \x/\x_{n}]| \leq |\Sigma| + 1$
\end{itemize}
and $\RK{\Pi} \leq |\M[\x/\x_{1}, ...\ , \x/\x_{n}]| \leq |\Pi|$. 

\item The most interesting case is for $\Pi$ of the shape
\small
$$\infer[(\wedge_n)]{\Pi \dem \bigwedge_{i=1}^n \Gamma_i \der \M: \stra{\sigma}{n}}{\Sigma_1 \dem \Gamma_1 \der \M: \sigma_{1} & ... & \Sigma_n \dem \Gamma_n \der \M: \sigma_{n}}$$
\normalsize
By inductive hypothesis, $\W{\Sigma_i}{r} \leq r^{\D{\Sigma_i}} \cdot \W{\Sigma_i}{1}$ for $1 \leq i \leq n$, and in particular $\max_{i=1}^n \W{\Sigma_i}{r} \leq r^{\max_{i=1}^n \D{\Sigma_i}} \cdot \max_{i=1}^n \W{\Sigma_i}{1}$.
Moreover, by Definition \ref{def:measures}, $\W{\Pi}{r} = r \cdot \max_{i=1}^n \W{\Sigma_i}{r}$, $\D{\Pi} = \max_{i=1}^n \D{\Sigma_i} + 1$ and $\W{\Pi}{1} = 1 \cdot \max_{i=1}^n \W{\Sigma_i}{1}$, therefore
$$r \cdot \max_{i=1}^n \W{\Sigma_i}{r} \leq r \cdot r^{\max_{i=1}^n \D{\Sigma_i}} \cdot \max_{i=1}^n \W{\Sigma_i}{1} = r^{\max_{i=1}^n \D{\Sigma_i} + 1} \cdot \max_{i=1}^n \W{\Sigma_i}{1}$$

and $\W{\Pi}{r} \leq r^{\D{\Pi}} \cdot \W{\Pi}{1}$.

\item We only show the case where $\Pi$ is of the shape
\small
$$\infer[(\wedge_n)]{\Pi \dem \bigwedge_{i=1}^n \Gamma_i \der \M: \stra{\sigma}{n}}{\Sigma_1 \dem \Gamma_1 \der \M: \sigma_{1} & ... & \Sigma_n \dem \Gamma_n \der \M: \sigma_{n}}$$
\normalsize
By inductive hypothesis $\W{\Sigma_i}{1} = |\M|$ for $1 \leq i \leq n$. Moreover, by Definition \ref{def:measures}, $\W{\Pi}{1} = 1 \cdot \max_{i=1}^n \W{\Sigma_i}{1}$, therefore $\max_{i=1}^n \W{\Sigma_i}{1} = |\M|$, and $\W{\Pi}{1} = |\M|$.

\end{enumerate}
\end{proof}

So we can give the following weighted version of Lemma \ref{lem:subs}:

\begin{lemma}[Weighted substitution] Let $\Pi \dem \Gamma, \x : \sigma \der \M : \tau$ and $\Sigma \dem \Delta \der \N : \sigma$, with $\Gamma \# \Delta$ and $\x\not\in dom(\Delta)$.
Then $S(\Sigma, \Pi) \dem \Gamma, \Delta \der \M[\N/\x] : \tau$ and $\W{S(\Sigma, \Pi)}{r} \leq \W{\Pi}{r} + \W{\Sigma}{r}$,  for every $r \geq \max \{ \RK{\Pi}, \RK{\Sigma} \}$.
\label{lem:wsubs}
\end{lemma}

\begin{proof}
By induction on the shape of $\Pi$: we will refer to the proof for Lemma \ref{lem:subs} and show that the condition on the measure holds. Again, most cases are trivial so we will only show the most meaningful ones, namely $(w)$, $(\wedge)$ and $(m)$.

If the last applied rule is $(w)$, since $\W{S(\Sigma, \Pi)}{r} = \W{\Pi'}{r} = \W{\Pi}{r}$, the inequality $\W{S(\Sigma, \Pi)}{r} \leq \W{\Sigma}{r} + \W{\Pi}{r}$ is satisfied: in fact, the sequence of rules needed to recover $\Delta$ is a sequence of $(w)$ and $(m)$ rules, which do not contribute to the weigth.

If the last applied rule is $(\wedge_n)$, with $n > 1$, then, by the proof for Lemma \ref{lem:subs},
$\W{\Pi}{r} = r \cdot \max_{i=1}^{n} \W{\Pi_{i}}{r}$
and
$\W{\Sigma}{r} = r \cdot \max_{i=1}^{n} \W{\sigma_{i}}{r}$.
By inductive hypothesis $S(\Sigma_{i}, \Pi_{i}) \dem \Gamma_i, \Delta_i \der \sub{\M}{\N'}{\x} : \tau_i$ and $\W{S(\Sigma_{i}, \Pi_{i})}{r} \leq \W{\Sigma_{i}}{r} + \W{\Pi_{i}}{r}$ for $1 \leq i \leq n$.
Since $\W{S(\Sigma, \Pi)}{r} = r \cdot \max_{i=1}^n \W{S(\Sigma_{i}, \Pi_{i})}{r} \leq r \cdot \max_{i=1}^n \W{\Sigma_{i}}{r} + r \cdot \max_{i=1}^n \W{\Pi_{i}}{r} = \W{\Sigma}{r} + \W{\Pi}{r}$, the inequality is satisfied.

Let the last applied rule be $(m)$, and let $k$ be its rank.
From the proof for Lemma \ref{lem:subs}, we can assume $\{ \x_1, ...\ , \x_k \} = \FV{\M} \cap \{ \x_1, ...\ , \x_n \}$, and moreover $\W{\Sigma}{r} = r \cdot \max_{i=1}^n \W{\Sigma_i}{r} =  r \cdot \max_{i=1}^n \W{\Sigma'_i}{r}$.

Let $r \geq \max \{ \RK{\Pi}, \RK{\Sigma} \} \geq \max \{ \RK{\Sigma_1}, ...\ , \RK{\Sigma_n}, \RK{\Pi'}, k \}$.
By induction $\W{S(\Sigma'_{1}, \Pi')}{r} \leq \W{\Sigma'_1}{r} + \W{\Pi'}{r}$;
then $\W{S(\Sigma'_{2}, S(\Sigma'_{1}, \Pi'))_2}{r} \leq \W{\Sigma'_2}{r} + \W{S(\Sigma'_{1}, \Pi')}{r}$, and so on.
By applying substitutions from $1$ through $k$ we get $\W{S(\Sigma'_{k}, ... S(\Sigma'_{1}, \Pi'))}{r} \leq \W{\Sigma'_k}{r} + \W{S(\Sigma'_{k-1}, ... S(\Sigma'_{1}, \Pi'))}{r}$.
By applying sequences of rules $\rho$ and $\delta$ to $S(\Sigma'_{k}, ... S(\Sigma'_{1}, \Pi'))$, and then a suitable sequence of $(w)$ and $(m)$ rules to recover the context $\Delta$, we get the desired proof. Notice that both $\rho$ and $\delta$ do not contribute to the weight.
By induction, $\W{S(\Sigma'_{k}, ... S(\Sigma'_{1}, \Pi'))}{r} \leq \W{\Sigma'_1}{r} + ... + \W{\Sigma'_k}{r} + \W{\Pi'}{r} \leq k \cdot \max_{i=1}^k \W{\Sigma'_i}{r} + \W{\Pi'}{r}$.
Since $\W{S(\Sigma, \Pi)}{r} = \W{S(\Sigma'_{k}, ... S(\Sigma'_{1}, \Pi'))}{r}$ and $\W{\Pi}{r} = \W{\Pi'}{r}$,
$$\W{S(\Sigma, \Pi)}{r} = \W{S(\Sigma'_{k}, ... S(\Sigma'_{1}, \Pi'))}{r} \leq k \cdot \max_{i=1}^k \W{\Sigma'_i}{r} + \W{\Pi'}{r} \leq
\W{\Sigma}{r} + \W{\Pi}{r}$$
and the inequality is satisfied.

\end{proof}


Using the previous property, we can prove that the weight of a proof decreases while reducing the subject.
\begin{lemma}
$\Pi \dem \Gamma \der \M: \sigma$ and $\M \redbeta \M'$ imply there is a derivation $\Pi' \dem \Gamma \der \M' : \sigma$, such that for every $r \geq \RK{\Pi}$, $\W{\Pi'}{r} < \W{\Pi}{r}$.
\label{lem:wnorm}
\end{lemma}

\begin{proof} As in the proof of Property \ref{prop:subjred}, we consider just the base case, when $\M=(\lambda \x.\Q)\N$. Then the most difficult case is when $\Pi$ ends by a non empty intersection tree. We will use the same terminology as in Property
\ref{prop:subjred}. Remember that $\Pi'$ is obtained from $\Pi$ by replacing every subproof $\Pi_i$:
\small
$$\infer[(\lin E)]{\Gamma_i, \Delta_i \der (\lambda \x. \Q_i) \N_i: \A_i}{\infer=[(\delta_i)]{\Gamma_i \der \lambda \x. \Q_i : \sigma_i \lin \A}{\infer[(\lin I)]{\Gamma'_i \der \lambda \x. \Q'_i : \sigma_i \lin \A_i}{\Sigma'_i\dem \Gamma'_i, \x: \sigma_i \der \Q'_i : \A_i}} & \Sigma''_i \dem \Delta_i \der \N_i : \sigma_i}$$
\normalsize
by $S(\Sigma''_i, \Sigma'_i)$, and leaving the intersection tree connecting all these subproofs unchanged.

By Lemma \ref{lem:wsubs}, for every $r \geq \RK{\Pi}$, $\W{S(\Sigma''_i, \Sigma'_i)}{r} \leq \W{\Sigma''_i}{r}+\W{\Sigma'_i}{r}$.

Since $\W{\Pi_i}{r}= \W{\Sigma''_i}{r}+\W{\Sigma'_i}{r} +1$, the proof is given.

%
%
%
%
%
%

\end{proof}

We can now prove that both the number of normalization steps and the size of the normal form are bounded by a function of the size of the term.

\begin{theorem}[Measure of reduction] Let $\Pi \dem \Gamma \der \M: \sigma$, and let $\M$ $\beta$-reduce to $\M'$ in $n$ steps. Then:
\begin{enumerate}[i)]
\item $n < |\M|^{\D{\Pi} + 1};$
\item $|\M'| < |\M|^{\D{\Pi} + 1}.$
\end{enumerate}
\label{th:redm}
\end{theorem}

\begin{proof} 
Let $\M \redbeta \M_1 \redbeta...\redbeta \M_n =\M'$. Then, by repeatedly applying 
Lemma \ref{lem:wnorm}, there is $\Pi_i \dem \M_i: \sigma$ such that, for all $r \geq \RK{\Pi}$,
$\W{\Pi_{i+1}}{r}<\W{\Pi_i}{r}$, for all $1\leq i\leq n-1$. Since the rank of a proof never increases when reducing the subject, if $r = \RK{\Pi}$, then $\W{\Pi_{i+1}}{r}<\W{\Pi_i}{r}$. Then the proof of the first point follows.

%
%
%
%


By Lemma \ref{lem:3m}, for all $i$, $|\M_i| = \W{\Pi_i}{1}$. Then the proof follows again from Lemma \ref{lem:wnorm}.


\end{proof}

So the exponent of the function is, in general, dependent on the term; for this reason, the bound on the normalization procedure can easily become exponential. Nevertheless, the proof given above is independend on a given reduction strategy.

\begin{remark}
One of the referees of this paper asked why we chosed the intersection as $n$-ary instead than binary connective, since by the lack of associativity the typability power of the system is the same in both cases, and binary intersection is more "standard". The answer is simple. We are interested not only in typability, but in using derivations for measuring the complexity of the reduction. Consider the term $M= (\lambda x y. y \underbrace{xx...x}_n)(II)$, where $I=\lambda x.x$. $|M|= 2n+6$. The minimal depth of a derivation in $\STI$ typing $M$ has depth $1$, and rank $n$, so the resulting bound for the number of $\beta$-reduction steps is $(2n +6)$, while the effective number of reductions is $2n +1$. In case of binary intersection, so modifying \STI\ in order to have only a multiplexor of rank $2$, the minimal derivation has depth $n-1$, and the resulting bound is $(2n+6)^n$, so becomes exponential.
\end{remark}
\medskip
